\documentclass{llncs}
\usepackage{comment}
\usepackage{graphicx}
\usepackage{indentfirst}
\usepackage{amsfonts}
\usepackage{latexsym,amsmath,epsfig,amssymb,graphics}
\usepackage{color}
\usepackage[table]{xcolor}
\usepackage{listings}
\usepackage{stmaryrd}
\usepackage{caption}
\usepackage[ruled,vlined,linesnumbered]{algorithm2e}
\usepackage[flushleft]{threeparttable}

\newcommand{\RRR}{\mathbb{R}}
\newcommand{\RR}{\mathbb{R}^+}
\newcommand{\TT}{\mathbb{T}}
\newcommand{\PP}{\mathbb{P}}

\newcommand{\NN}{\mathbb{N}}
\newcommand{\ZZ}{\mathbb{Z}}
\newcommand{\TA}{\mathcal{A}}

\newcommand{\pta}{{\em parametric timed automata}}

\newcommand{\LTS}{\mathcal{L}}

\newcommand{\CAD}{{\tt CAD}}

\newcommand{\linf}{{\tt linf}}
\newcommand{\usup}{{\tt usup}}
\newcommand{\CF}{{\tt cf}}

\newcommand{\EUP}{{\tt up}}
\newcommand{\ELB}{{\tt lb}}

\newcommand{\aaa}{{\vec{\alpha}}}

\newcommand{\sgn}{{\tt sgn}}
\newcommand {\signs}{{\tt signs}}
\newcommand{\oomit}[1]{}

\newcommand{\Model}[1]{[\![#1]\!]}
\newcommand{\pcc}{{\em parametrically constrained clock}}

\newcommand{\deff}{\  \widehat{=}\ }

\graphicspath{{img/}}

\title{Parameter Synthesis Problems for  one parametric
	clock  Timed Automata}

\author{Liyun Dai\inst{1}\thanks{Corresponding author} \and Taolue Chen\inst{2} \and Zhiming Liu\inst{1} \and Bican Xia\inst{3} \and Naijun Zhan\inst{4} \and Kim G. Larsen\inst{5}}
\institute{RISE, Southwest University, Chongqing, China 
	\and Department of Computer Science and Informatioin Systems, \\ 
		Birkbeck, University of London, UK \and
	LMAM \& School of Mathematical Sciences, Peking University \and
	State Key Laboratory of Computer Science, Institute of Software, CAS \and
	CISS, CS, Aalborg University, Denmark\\
	\email{dailiyun@swu.edu.cn ~ taolue@dcs.bbk.ac.uk ~zhimingliu88@swu.edu.cn ~ xbc@math.pku.edu.cn ~ znj@ios.ac.cn ~ kgl@cs.aau.dk}
}
\date{}
\begin{document}

\maketitle
%
%
%
%

\begin{abstract}
	In this paper, we study the parameter synthesis problem for a class of
	 parametric timed automata. The problem asks to construct the set of 
	valuations of the parameters in the parametric timed automaton,
	 referred to as the feasible region, under which the resulting 
	timed automaton satisfies certain properties. We show that the parameter
	synthesis problem of parametric timed automata with only one parametric
	 clock (unlimited  concretely constrained clock) and arbitrarily many  
	parameters is solvable when all the expressions are linear expressions. And it is moreover 
	the synthesis problem is solvable when the form of constraints 
	are parameter polynomial inequality not just simple constraint and parameter domain is nonnegative real number.  

	 
\end{abstract}

\section{Introduction}
\label{sec:inr}
Real-time applications are  increasing importance, so are  their complexity and requirements for trustworthiness,  in the era of Internet of Things (IoT), 
especially  in the areas of industrial control and smart homes. Consider, for example, the control system of a boiler used in house. Such a system is required 
to  switch  on the gas within a certain bounded period of time  when the water gets too cold.  Indeed, the design and implementation of the system  not only 
have to  guarantee the correctness of system functionalities,  but also need to assure  that the application is in compliance with the non-functional 
requirements, that are timing constraints  in this case. 

{\em Timed automata} (TAs)  \cite{Alur90,alur1994a} are widely used  for modeling and verification of real-time systems. However, one disadvantage of the 
TA-based approach is that  it  can only be used to  verify \emph{concrete} properties, i.e., properties with concrete values of all timing parameters occurring in 
the system. Typical examples of such  parameters  are  upper and lower bounds of computation time, message delay and time-out. This makes the traditional 
TA-based approach  not ideal for the design of real-time applications because in the \emph{design phase}  concrete values are often  not available. This 
problem is usually dealt with extensive trial-and-error and prototyping activities to find out what concrete values of the parameters are suitable. This 
 approach of design is costly, laborious, and error-prone, for at least two reasons:  
(1) many trials with different parameter configurations suffer from unaffordable costs, without enough assurance  of a safety standard because a sufficient 
coverage of configurations is  difficult to achieve; (2) little  or no feedback information is provided to the developers to help improve the design when  a system 
malfunction is detected.

\subsection{Decidable parametric timed automata}
To mitigate the limitations of the TA-based approach, 
{\em parametric timed automata} (PTAs) are  proposed \cite{alur1993parametric,annichini2000symbolic,bandini2001application,HUNE2002183}, which  allow 
more general constraints on  invariants of notes (or states) and guards of edges (or transitions) of an automaton. Informally, a  clock $x$ of a PTA $\TA$  is 
called a \pcc\ if $x$ and some parameters both occur in a constraint  of $\TA$.    Obviously, given any  valuation of the  parameters in a PTA, we obtain a 
concrete TA. One of the most important questions of PTAs is the \emph{ parameter synthesis problem}, that  is, for a given property to compute the entire set of valuations 
of the parameters for a PTA such that when the parameters are instantiated by these valuations, the resulting TAs all satisfy the property. The  synthesis 
problem for general PTAs is known to be undecidable. There are, however, several proposals  to restrict the general PTAs from different perspectives to gain 
decidability.   Two kinds of restrictions that are being widely investigated are (1) on  the number of clocks/parameters in the PTA;  and (2) on the way in which  
parameters are bounded, such as the  L/U PTAs \cite{HUNE2002183}.

  There are many works about \pta. An algorithm  based on backward to solve  nontrivial class of parametric verification problems is   presented in 
  \cite{alur1993parametric}.  The authors  have proved that a large class of parametric verification problems are undecidable; they have also showed 
    that the remaining (intermediate) class of parametric verification problems for which  then have neither decision procedures nor undecidability results are 
    closely related 
    to various hard and open problems of logic and automata theory. A semi-algorithm approach based on 
    (1) expressive symbolic representation structures is called 
    parametric DBP's, and (2) accurate extrapolation techniques allow to speed up the reachability analysis and help its termination is  proposed in 
    \cite{annichini2000symbolic}.   An algorithm and the tools for reachability analysis of hybrid systems  is presented	in \cite{alur2002reachability}. They 
    combine the notion of predicate abstraction  
    with resent techniques for approximating the set of reachable states of linear systems using polyhedron. The main diffcult of this method is how to find the 
    enough predicates.   
    In \cite{jovanovic2015integer}, the authors give a method without  an explicit enumeration to  synthesize
all the values of parameters and give symbolic algorithms for reachability and unavoidability properties. 
  An adaptation of counterexample guided abstraction refinement (CEGAR) with which  one can obtain an under approximation of the set of good parameters 
  using linear programming   is proposed in \cite{frehse2008counterexample}.
   An inverse method which synthesizes the constraint of parameters for an existing trace such that it can guarantee its
    executes of \pta\ under this constraint with same previous trace is  provided  in \cite{andre2009inverse}. In \cite{jovanovic2015integer}, the authors 
    provide a subclass of parametric timed automata which they can actually and efficiently analyze. The author 
    of \cite{Andr16} makes a survey of decision and computation problems progress based on the recent 25 years' researches on these problems.

    The constraints in above works are simple constraint which means that in the form of constraint as $x\prec c$ ($x-y\prec c$), $x\prec p$ ($x-y\prec p$) or 
    logical combination of above forms
    where $x,y$ are clocks, $c$ is a constant and $p$ is parameter.  In this paper, we will extended  the form to $x\prec f(p_1,\cdots,p_m)$ 
    ($x-y\prec f(p_1,\cdots,p_m)$) where $p_1,\cdots, p_m$ are parameters and $f$ is a polynomial in $\ZZ[p_1,\cdots,p_m]$. 
    
    There are many works related to solving polynomial constraints problems e.g. \cite{Tarski51,collins1}.
    
    As one would expect, Tarski’s procedure consequently has been much im-
    proved. Most notably, Collins \cite{collins1} gave the first relatively effective method of
    quantifier elimination by cylindrical algebraic decomposition (\CAD). The \CAD\
    procedure itself has gone through many revisions \cite{Caviness1998,Hong,McCallum1,McCallum2,brown,Dolzmann2004,han2014constructing}.
    The \CAD\ algorithm works by
    decomposing $\mathbb{R}^k$ into connected components such that, in each cell, all of the
    polynomials from the problem are sign-invariant. To be able to perform such
    a particular decomposition, \CAD\ first performs a projection of the polynomials
     from the initial problem. This projection includes many new polynomials,
    derived from the initial ones, and these polynomials carry enough information
    to ensure that the decomposition is indeed possible. Unfortunately, the size of
    these projections sets grows exponentially in the number of variables, causing
    the projection phase to be a key hurdle to \CAD\ scalability.

    \subsubsection{Contribution}
    	In this paper, we study the parameter synthesis problem
    	of a class of parametric time automata.
    	We show that the parameter
    	synthesis problem of parametric timed automata with only one parametric
    	clock (unlimited  concretely constrained clock) and arbitrarily many  
    	parameters is solvable when all the expressions are linear expressions. And it is moreover 
    	the synthesis problem is solvable when the form of constraints 
    	are parameter polynomial inequality  and parameter domain is nonnegative real number.  
    
  {\small		
    	\begin{table}
    		\begin{threeparttable}
    		\caption{Our PTA results}
    	\begin{tabular}{|c|c|c|c|c|c|c|c|}    
    		 \hline
    		 $\TT$ & $\PP$& Constraints & P-clocks & NP-clocks & Params  & emptiness & synthesis\\
    		  \hline 
    		 $\NN$ & $\RRR$ & Polynomial constraints  &      1       &  0           &    any  &    ~  & solvable \\
    		  \hline 
    		  $\NN$ & $\ZZ$ & Simple constraints  &      1       &  any           &    any  &    ~  & solvable \\
    	     	  \hline 
    	\end{tabular}
    	 \begin{tablenotes}
    	 	\item  ``$\TT$" to denote the domain of clock.
    	 	\item  ``$\PP$" to denote the domain of parameter.
    	 	\item ``Constraints" is form of constraint in PTA include constraints occurring in  property.
    	 	\item ``P-clocks" is the number of parametric clock.
    	 	\item ``NP-clocks" is the number of concretely constrained clock.
    	 	\item ``Params" is the number of parameters occurring in PTA.
    	 	\item ``emptiness" denote the whether  decidable of emptiness problem.
    	 	\item ``synthesis" denote the whether  decidable of synthesis problem. 
    	 \end{tablenotes}
    	\end{threeparttable}
    	\end{table}
    	
}
 \subsubsection{Related work}
  Besides the above mentioned works, there are several other results that related to ours.   The idea of limiting the number of parameters used such that upper and lower bounds cannot share a same parameter is also presented in \cite{alur2001parametric} where the authors studied the logic 
  {\tt LTL} augmented with parameters. And  our topic \pta\ is different from theirs. An extension of the model checker {\tt UPPAAL} presented in \cite{HUNE2002183} is capable of synthesizing linear parameter constraints for the correctness of \pta\ and it also  identifies  a subclass of \pta\  (L/U automata) for which the emptiness problem is decidable.  Decidability results  for L/U automata have been further investigated  in \cite{bozzelli2009decision} where   the constrained versions of emptiness and universality of the set of parameter valuations for
  which there is a corresponding infinite accepting run of the automaton is studied and decidability if parameters of different types (lower and upper bound parameters) are not compared in the linear constraint is obtained. They show how to compute the explicit representation of the set of parameters when all the parameters are of the same type (L-automata and U-automata).  Compared with~\cite{bozzelli2009decision} which considers liveness problems of the system, our results are related to synthesis parameter which satisfies a given property. In \cite{bruyere2007realtime}, the authors show that
  the model-checking problem is decidable and the parameter synthesis problem
  is solvable, in discrete time, over a PTA with one parametric clock, if equality
  is not allowed in the formula. Compared with it, we  do not have  equality restriction. In \cite{andre2015language}, the authors proved that 
  the language-preservation problem is decidable for deterministic for the \pta\ with all lower bound parameters or all upper bound parameters and one parameter. However, the limitations we consider for obtaining decidability is orthogonal to those presented in \cite{andre2015language}. 
  In \cite{bundala2014advances}, the authors prove that the emptiness problem of \pta\ with two parameter clocks and one parameter is decidable.

 \subsubsection{Organization}
 After the introduction, the definition of \pta\ is presented in Section \ref{sec:pre}. In Section \ref{sec:newresult1} 
some theoretical results about  parameter synthesis problem are given. Based on result of \CAD\ we prove that
with only one parametric
clock and arbitrarily many parameters is solvable. And it is moreover
the form of constraints are parameter polynomial inequality.   In Section \ref{sec:simple},  We show that the parameter
synthesis problem of parametric timed automata with only one parametric
clock (unlimited  concretely constrained clock) and arbitrarily many  
parameters is solvable  when all the expressions are linear expressions.


\section{Parametric Timed Automata}
\label{sec:pre}
We introduce the basis of  PTAs  and set up terminology for our discussion. We first define some preliminary notations before we introduce PTAs. 
We will use a model of  labeled transition systems (LTS) to define  semantic behavior of PTAs.

\subsection{Preliminaries}

We use $\ZZ$, $\NN$, $\mathbb{R}$ and $\RR$ to denote the  sets of integers, natural numbers, real numbers and non-negative real numbers, respectively.  
Although each  PTA involves only a finite number of clocks and a finite number parameters, we need an infinite set  of {\em clock variables} (also simply 
called {\em clocks}), denoted by $\mathcal{X}$ and an infinite set of {\em parameters}, denoted by $\mathcal{P}$, both are enumerable. We use $X$ and $P$
 to denote (finite) sets of clocks and parameters and   $x$ and $p$, with subscripts if necessary,  to denote clocks and parameters, respectively.  We use $\TT$ 
 to denote the domain of clocks. We are mostly interested in the case that  $\TT=\NN$ or  $\TT=\RR$ of nonnegative reals. Unless explicitly 
specified, our results are applicable in either case. We use $\PP$ to denote the domain of clocks. We are mostly interested in the case that 
$\PP=\ZZ$ or  $\PP=\RRR$.

We mainly consider dense time, and thus we  define a {\em clock valuation} $\omega$
as a function of the type $\mathcal{X}\mapsto \TT$. For a finite set $X=\{x_1, \ldots, x_n\}$ of clocks, an evaluation $\omega$  restricted on $X$ can be represented by  
a $n$-dimensional   point $\omega(X)=(\omega(x_1), \omega(x_2),\ldots,\omega(x_n))$, and it is called an {\em parameter valuation of $X$} and 
simply denoted as $\omega$ when there is no confusion. Given a constant $d\in \TT$, we use $\omega+d$ to denote the evaluation that assigns any 
clock $x$ with the value $\omega(x)+d$, and $(\omega+d)(X)= (\omega(x_1)+d,\omega(x_2)+d,\ldots,\omega(x_n)+d)$. When $n=1$, we directly use $\omega$ as the value of clock $x_1$.  Similarly,  a 
{\em parameter valuation} $\gamma$  is an assignment of values to the  parameters, that is
 $v: \mathcal{P}\mapsto \PP$.  For a finite set $P=\{p_1,\ldots, p_m\}$ of $m$ parameters, a parameter valuation $\gamma$ restricted on $P$
  corresponds to  a $m$-dimensional point $(\gamma(p_1),\gamma(p_2),\ldots,\gamma(p_m))\in \PP^m$, and we use this vector to denote the 
  valuation $\gamma$ of $P$ when there is no confusion. When $m=1$, we directly use $\gamma$ as the value of $p_1$.

\begin{definition}[Expression]
	A linear expression $e$ is either an expression of the form $c_0+c_1p_1+\cdots+c_np_n$ where $c_0,\cdots,c_n\in \ZZ$, or $\infty$.
	We use $\CF(e,p)$ to denote the coefficient of $p$ in linear expression $e$.
	A polynomial expression is an expression of the form $\sum_{i=0}^h c_ip_1^{k_{i,1}}c\cdots p_m^{k_{i,m}}  $ where $c_0,\cdots,c_h\in \ZZ, k_{i,j}\in \NN$.
\end{definition}
We also write polynomial $f$ as form 
\[f(Y,x)=c_l'p_m^{d_l}+c_{l-1}'p_m^{d_{l-1}}+\cdots+c_1'p_m'^{d_1}+c_0'\]
where $Y=[y_1,\cdots,y_{k}],\ 0< d_1<\cdots < d_l $, and the coefficients $c_i'$ are in $\ZZ[y_1,\cdots, y_{k}]$ with $c_{l}'\neq 0$.
	
We use $\mathcal{LE}$ and $\mathcal{PE}$ to denote the set of linear expressions and polynomial expression, respectively. We use $\mathcal{E}$ to denote 
set $\mathcal{LE} \cup  \mathcal{PE}$. For an $e\in \mathcal{LE}$, we use  $\textit{con}(e)$  the constant $c_0$, and  $\CF(e,p)$  the coefficient of $p$ in $e$, i.e. $c_i$ if $p$ 
 is $p_i$ for $i=1,\ldots, m$, and $0$,  otherwise. For the convenience of discussion, we also say the infinity $\infty$ is a  expression.  We call  expression 
 $e$ a {\em parametric expression} if it contains  some parameter, a {\em concrete  expression}, otherwise (i.e., $e$ is parameter free).

A PTA only allows {\em parametric constraints} of the form  $x-y\sim e$, where $x$ and $y$ are clocks, $e$ is an  expression, and the  ordering  relation
 $\sim\in  \{>,\ge, <,\le, =\}$. A constraint $g$ is called a  {\em parameter-free} (or {\em concrete}) {\em  constraint} if the expression in it is concrete.  
 For an expression $e$, a parameter valuation $\gamma$, a clock valuation $\omega$ and a constraint $g$, let 
\begin{itemize}
\item $e[\gamma]$ be the (concretized) expression obtained from $e$ by substituting the value $\gamma(p_i)$ for $p_i$ in $e$, i.e. when $e$ is a
 linear expression $e=c_0+c_1p_1+\cdots, c_mp_m$, then $e[\gamma]=c_0+ c_1\times \gamma(p_1)+ \ldots + c_m\times \gamma(p_m)$,
\item $g[\gamma]$  be the predicate obtained from constraint $g$  by substituting the value $\gamma(p_i)$ for $p_i$ in $g$, and 
\item $\omega \models g$ holds if $g[\omega]$  holds.
\end{itemize}

A pair $(\gamma, \omega )$ of parameter valuation and clock valuation   gives an evaluation to any parametric constraint $g$. We use $g[\gamma, \omega]$
 to denote the truth value  of $g$ obtained by substituting each parameter $p$  and each clock $x$  by their values $\gamma(p)$ and $\omega(x)$, 
 respectively. We say the pair of valuations $(\gamma, \omega)$ satisfies constraint $g$, denoted by $(\gamma, \omega)\models g$,  if $g[ \gamma, \omega]$
  is evaluated to true. For a given parameter valuation $\gamma$, we define  $\Model{g[\gamma]}= \{\omega\mid (\gamma, \omega)\models g\}$ to be the set 
  of clock valuations which together with $\gamma$ satisfy $g$.

A clock $x$ is reset by an {\em update} which is an expression of the form $x:= b$, where  $b\in \NN$. Any reset $x:=b$ will change a clock valuation $
\omega$ to a clock valuation $\omega'$ such that $\omega'(x) =b$ and $\omega'(y) =\omega(y)$  for any other clock $y$. Given a clock valuation $\omega$ 
and a set $u$ of updates, called an {\em update set}, which contains at most one reset for one  clock,  we use $\omega[u]$ to denote the clock valuation after 
 applying all the clock resets in $u$ to $\omega$.  We  use $c[u]$ to denote the constraint which is used to assert the relation of the parameters with the clocks 
 values after the clock resets of $u$. Formally,  $c[u](\omega)\deff c(\omega[u])$ for every clock valuation $\omega$. 

It is easy to see that the general constraints $x-y\sim e$ can be expressed in terms of {\em atomic constraints} of the form $b_1x-b_2y\prec e$,  where 
$\prec \in \{<,\le\}$ and $b_1,b_2\in \{0,1\},e\in \mathcal{E} $. To be explicit, an atomic constraint is in one of the following three  forms $x-y \prec e$, $x\prec e$, 
or $-x\prec e$.  We can write  $-x_i \prec  e$ as $x_i\succ -e$. and $x-y\prec e$ as $y-x\succ -e$, where $\succ \in \{>,\geq\}$.  However, in this paper we mainly  consider  {\em simple 
constraints}  that are  finite  conjunctions of atomic constraints. 

\subsection{Parametric timed automata}
We assume the knowledge of timed automata (TAs), e.g., \cite{alur1999timed,bengtsson2004timed}. A  clock constraint of a TA  either a  {\em invariant 
property} when the TA is  in a state (or location)  or a {\em guard condition} to enable the changes of states (or a state transition). Such a constraint  is  in 
general  a Boolean expression of parametric free  atomic constraints. However, we can assume that the  guards and invariants of TA are simple concrete  
constraints, i.e. conjunctions of concrete atomic constraints. This is because we can always transform a TA  with  disjunctive guards and invariants to an 
equivalent  TA  with  guards and invariants which are simple constraints only.   

In what follows, we define  PTAs which extend TAs  to allow the use of parametric simple constraints as guards and invariants (see \cite{alur1993parametric}).

\begin{definition}[PTA]
	\label{def:pta}
	Given a finite set of clocks $X$ and a finite set of parameters $P$, a PTA   is a 5-tuple $\TA=(\Sigma,Q,q_0,I,\rightarrow)$,
	where
	\begin{itemize}
		\item $\Sigma$ is a finite set of actions.
		\item  $Q$ is a finite set of locations and  $q_0\in Q$ is  called the initial location,
		\item  $I$ is the invariant, assigning to every $q\in Q$ a simple constraint $I_q$ over 
		the clocks $X$ and parameters $P$, and
		\item $\rightarrow$
		is a discrete transition relation whose elements are of the form $(q,g,a,u,q')$, where $q,q'\in Q$, $u$ is an update set, $a\in \Sigma$ and $g$ 
		 is a simple constraint.
	\end{itemize}
\end{definition}
Given a PTA $\mathcal{A}$, a tuple $(q,g,a,u,q')\in \rightarrow$ is also denoted  by $q \xrightarrow{g\&a[u]} q'$, and it is called a transition step (by the 
guarded action $g\&a$). In this step, $a$ is the  action that triggers the transition. The constraint $g$ in the transition step is called the {\em guard} of the 
transition step, and only when $g$ holds in  a location  can the transition  take place. By this transition step, the system modeled by the automaton changes 
from location $q$ to location $q'$, and the clocks are reset by the updates in $u$. However, the meaning of the guards and clock resets and acceptable runs
 of a PTA will be defined by a labeled transition system (LTS) later on. At this moment, we define a {\em syntactic run} of a PTA $\mathcal{A}$ as a sequence 
 of consecutive transitions step starting from the initial location
\[\tau = (q_0, I_{q_0})\xrightarrow{g_1\&a_1[u_1]}(q_1, I_{q_1})\cdots \xrightarrow{g_\ell \&a_\ell [u_\ell]}(q_\ell, I_{q_\ell}).
\] We call a syntactic run $\tau$ is a {\em simple syntactic run} if $\tau$ has no location variants and clock resets.
 
Given a PTA $\TA$, a clock $x$ is said to be a {\em parametrically constrained  clock} in $\TA$  if there is a parametric constraint containing $x$. 
Otherwise,   $x$ is a concretely   constrained  clock. We can follow the procedures in   \cite{alur1993parametric} and \cite{bundala2014advances} to  
eliminate  from  $\TA$ all the concretely constrained clocks.  Thus, the  rest of this paper only considers the PTAs  in which all clocks are parametrically 
constrained. We use $\textit{expr}(\mathcal{A})$ and $\textit{para}(\mathcal{A})$ to denote the set of all  expressions and parameters in a PTA
 $\mathcal{A}$, respectively.
 
\begin{figure}[h!]
	\centering
	\includegraphics[width=0.65\textwidth]{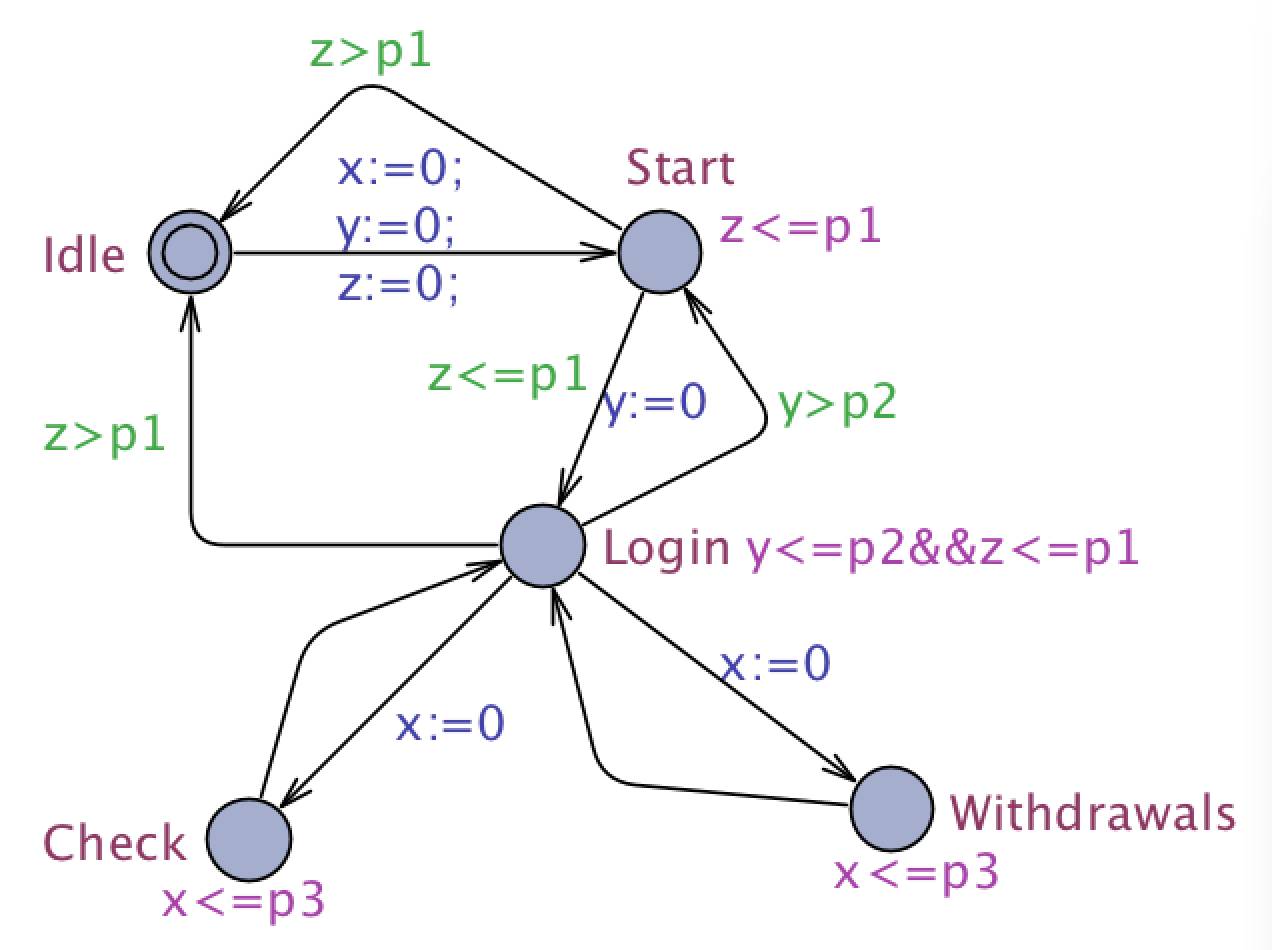}
	\caption{An ATM modeled using a PTA.} \label{fig:atm}
\end{figure}

\subsubsection*{Example 1}
The PTA in {\em Fig. \ref{fig:atm}} models an ATM. It has 5 locations, 3 clocks $\{x,y, z\}$  and 3 parameters $\{p_1,p_2, p_3\}$. This PTA is deterministic and
all the clocks are parametric. To understand the behavior of state transitions,  for examples, the  machine can initially idle for an arbitrarily long time. Then, 
 the user can start the system by, say, pressing a button and the PTA enters location ``Start" and resets the three clocks. The machine can remain in ``Start" 
 location as long as the invariant  $z\le p_1$ holds, and during this  time  the user can drive the system (by pressing a corresponding button) to login their 
 account and the automaton  enters  location ``Login" and resets clock $y$.  A time-out action occurs and it  goes back to ``Idle'' if  the machine stays at 
 ``Start''  for too long and the  invariant $z\leq p_1$ becomes false. Similarly, the machine can remain in location ``Login" as long as the invariant
  $y\le p_2 \wedge z\le p_1$ holds and during this time the user can decide either to ``Check'' (her balance) or to ``Withdraw" (money), say by pressing  
  corresponding buttons.   However,  if the user does not take any of these actions $p_2$ time units after the machine enter location ``Login",  the machine
   will back  to ``Start" location.

\subsection{Semantics of PTA via labeled transition systems}
We use a  standard model of  {\em labeled transition systems} (LTS) for describing  and analyzing the behavioral properties of  PTA.
\begin{definition}[LTS]
	\label{def:lts} A \emph{labeled transition system} (LTS) over a set of (action) symbols $\Delta$ is a triple
	$\LTS=(S,S_0,\rightarrow)$, where
	\begin{itemize}
	\item  $S$ is a set of states with a subset  $S_0\subseteq S$ of states called the  initial states.
	\item $\rightarrow \subseteq S\times \Delta  \times S$  is a relation, called the transition relation. 
	\end{itemize}
	We write $s \xrightarrow{a} s' $ for a triple $(s,a,s') \in \rightarrow$ and it is called a transition step by action $a$. 
	
	A run  of $\LTS$ is a finite alternating sequence of states in $ S$ and actions  $\Delta$, $ \xi = s_0a_1s_1\ldots a_\ell  s_\ell $,   such that 
	$s_0\in S_0$ and $s_{i-1}\xrightarrow{a_i}s_i\in \rightarrow$  for $i=1,\ldots,\ell$. A run $\xi$ can be written in the form of 
	$s_0\xrightarrow{a_1} s_1 \xrightarrow{a_2} \cdots \xrightarrow{a_{\ell}} s_\ell $.  	
	The length of a run $\xi$ is its number $\ell$ of transitions steps and it is denoted as $|\xi|$, and a state $s\in S$ is called reachable in $\LTS$ if 
	$s$ is the last state a run of $\LTS$, e.g. $s_\ell$ of $\xi$.  
\end{definition}

\begin{definition}[LTS semantics of PTA]
	\label{def:sem}
	For a PTA $\TA =(\Sigma, Q, q_0, I,\rightarrow)$ and a parameter valuation $\gamma$, the concrete semantics of PTA under $\gamma$,
	 denoted by $\TA[\gamma]$,  is the LTS $(S,S_0,\rightarrow)$ over $\Sigma\cup \RR$, where
\begin{itemize}
\item a state in  $S$  is a  location $q$ of  $\mathcal{A}$ augmented with the clock valuations which together with the parameter valuation $\gamma$ 
satisfy the invariant   $I_q$ of the location, that is
	\[S=\{ (q,\omega)\in Q\times (X\rightarrow \RR)\mid (\gamma, \omega)\models I_q\}\]
  \[	S_0=\{(q_0,\omega)\mid (\gamma, \omega) \models I_{q_0} \wedge \omega=(0,\cdots,0)\} \]

\item  any transition step in the transition $\rightarrow$ of the LTS is either an instantaneous  transition step  by  an action in $\Sigma$ defined by 
$\mathcal{A}$ or by a time advance, that are specified by the following rules, respectively
	\begin{itemize}
		\item {\bf instantaneous transition}:  for any $a\in \Sigma$, $(q,\omega)\xrightarrow{a}(q',\omega') $ if  there are  simple constraint $g$ and an  
		update set  $u$ such that $q\xrightarrow{ g\&a[u]} q'$, $(\gamma,\omega) \models g $ and $\omega'=\omega[u]$; and 
		\item {\bf time advance transition} $(q,\omega) \xrightarrow{d} (q',\omega') $ if $q'=q$ and $\omega'=\omega+d$.
	\end{itemize}
  \end{itemize}
\end{definition}
A {\em concrete run} of a PTA $\mathcal{A}$ for a given valuation $\gamma$ is a sequence of consecutive state transition steps 
$\xi=s_0\xrightarrow{t_1} s_1 \xrightarrow{t_2} \cdots \xrightarrow{t_{\ell}} s_\ell $  of the LTS $\mathcal{A}[\gamma]$, which we also call a run of the LTS
 $\mathcal{A}[\gamma]$.  A state $s = (q,\omega)$ of $\mathcal{A}[\gamma]$ is a {\em reachable state} of  $\TA[\gamma]$  if there exists some run  
 $\xi=s_0\xrightarrow{t_1} s_1 \xrightarrow{t_2} \cdots \xrightarrow{t_{\ell}} s_\ell $ of  $\TA[\gamma]$ such that $s=s_\ell $.

Without the loss of generality, we merge  any two consecutive time advance transitions respectively labelled by  $d_i$ $d_{i+1}$ into a single time advance 
transition labels by $d_i+d_{i+1}$. We can further merger a consecutive pair $s\xrightarrow{d} s'\xrightarrow{a} s''$ of a timed advance transition by $d$ and  
an instantaneous transition by an action $a$  in a run into a single observable transition  step $s\xrightarrow{a} s''$. If we do this repeatedly until all time 
advance steps are eliminated, we obtain an {\em untimed run} of the PTA (and the LTS), and the sequence of actions in an untimed run is called a  {\em trace}.

We call an untimed run  $\xi=s_0\xrightarrow{a_1}s_1\cdots \xrightarrow{a_\ell}s_\ell$  a {\em simple run}  if $\omega_i\ge \omega_{i-1}$ for
  $i=1,\cdots, \ell$, where $s_i=(q_i,\omega_i)$. It is easy to see that $\xi$ is a {\em simple untimed run} if each transition by $a_i$ does not have any clock 
  reset in $\xi$. 

\begin{definition}[LTS of trace]
 For a   PTA  $\mathcal{A}$ and a syntactic run
 \[\tau=(q_0{,}I_{q_0})\xrightarrow{g_1\&a_1[u_1]}(q_1{,}I_{q_1}){\cdots} \xrightarrow{g_\ell{\&}a_\ell[u_\ell]}(q_\ell{,}I_{q_\ell})\]
we define the PTA $\TA_\tau{=}(\Sigma_\tau {,}Q_\tau {,}q_{0,\tau}{,}I_\tau,\rightarrow_\tau)$, where 
 \begin{itemize}
 \item $\Sigma_\tau =\{a_i \mid i= 1,\cdots, \ell\}$,
 \item $Q_\tau =\{q_0,\cdots, q_\ell\}$ and $q_{0,\tau}=q_0$,
 \item  $I_\tau (i)=I_{q_i}$ for $i\in Q$, and 
 \item $\rightarrow_\tau =\{(q_{i-1},g_i, a_i, u_i,q_i)\mid i=1,\cdots, \ell \}$.
 \end{itemize}
Give a parameter valuation $\gamma$, the  concrete semantics
of $\tau$ under $\gamma$ is  defined to be the LTS $\TA_\tau[\gamma]$. 
\end{definition}
For a syntactic run 
 \[\tau=(q_0{,}I_{q_0})\xrightarrow{g_1\&a_1[u_1]}(q_1{,}I_{q_1}){\cdots} \xrightarrow{g_\ell{\&}a_\ell[u_\ell]}(q_\ell{,}I_{q_\ell})\]
We use $R(\TA_\tau[\gamma])$ to denote the set of  states $(q_k, \omega_k)$ of  $\mathcal{A}_\tau[\gamma]$ such that the following is an
 untimed run of $\TA_\tau[\gamma]$
 \[\xi=(q_0,\omega_0)\xrightarrow{a_1} (q_1,\omega_1) \cdots \xrightarrow{a_k}(q_k, \omega_k) \cdots\xrightarrow{a_\ell}(q_\ell, \omega_\ell).\]
 We also call $\xi$ is a run of syntactic run $\tau$ under $\gamma$.
 We use $\Gamma(\TA_{\tau})$ to denote the entire set of parameter valuation $\gamma$ which makes  $R(\TA_\tau[\gamma])\neq \emptyset$.

\subsection{Two decision problems for PTA}
We first present the  properties of PTAs  which we consider in this paper. 

\begin{definition}[Properties]
	\label{def:prop}
	A  {\em state  property} and a {\em system property} for a PTA  are specified by a state predicate $\phi$ and a temporal formula $\psi$ defined by 
	the following syntax, respectively: for $x,y\in X$, $e\in \mathcal{E}$ and $\prec\in \{<,\le, =\}$  and $q$ is a location.
	
	\[\begin{array}{lllll}
	  \phi&::=& x\prec e \mid -x \prec e\mid x-y\prec e\mid  q\mid \neg \phi\mid \phi  \wedge \phi \mid \phi \vee \phi\\
	  \psi&::=&\forall \Box\phi \mid \exists \Diamond \phi
	  \end{array}
	  \]

\end{definition}

Let  $\gamma$ be a parameter valuation  and   $\phi$ be a state formula. We  say $\TA[\gamma]$ {\em satisfies} $\exists\Diamond \phi$, denoted  by $
\TA[\gamma]\models \exists\Diamond \phi$,  if there  is a  reachable state $s$ of  $\TA[\gamma]$  such that $\phi$ holds in  state $s$. Similarly,  
$\TA[\gamma]$ {\em satisfies}  $\forall \Box \phi$, denoted by $\TA[\gamma] \models \forall \Box \phi$,  if $\phi$ holds in all reachable states of $\TA[\gamma]$.
 We can see that if $\TA[\gamma]{\models}\exists\Diamond \phi$,  there is an syntactic  run $\tau$ such that there is  a state in  $R(\mathcal{A}_\tau[\gamma])$
   satisfies  $\phi$. In this case, we also say that the syntactic   run $\tau$ satisfies $\phi$ under the parameter valuation $\gamma$. We denote it by 
   $\tau[\gamma]\models \phi$.

We are now ready to present the formal statement of the parameter synthesis problem and the emptiness problem of PTA.

\begin{problem}[The parameter synthesis problem]
	\label{pro:synth}
	Given a PTA $\TA$ and a system property  $\psi$, compute the entire set  $\Gamma(\TA,\psi)$ of parameter valuations such that
	 $\TA[\gamma]\models \psi$ for each  $\gamma \in \Gamma(\TA,\psi)$.
\end{problem}

 Solutions to the problems are important in system plan and optimization design. Notice that when there are no parameters in $\mathcal{A}$,  the problem 
 is decidable in PSPACE \cite{alur1994a}. This  implies  that if there are parameters in $\mathcal{A}$, the satisfaction  problem $\TA[\gamma]\models \psi$ 
 is decidable in PSPACE for any given parameter valuation $\gamma$.

A special case of the synthesis problem is the emptiness problem, which is  by itself very important and formulated below.
\begin{problem}[Emptiness problem]
	\label{pro:emp}
Given a PTA  $\TA$ and a system property $\psi$, is there a parameter valuation $\gamma$  so that $\TA[\gamma]\models \psi$?
\end{problem}
This is equivalent to the problem of checking if the set $\Gamma(\mathcal{A}, \psi)$ of feasible parameter valuations is empty.

Many safety verification problems can be reduced to the emptiness problem. We say that {\em Problem \ref{pro:emp}} is a special case of 
 {\em Problem \ref{pro:synth}} because solving the latter for a PTA $\TA$ and a property $\psi$ solves {\em Problem \ref{pro:emp}}. 
 
It is known that  the emptiness problem is decidable for a PTA with only one clock \cite{alur1993parametric}. However, the problem becomes undecidable
 for PTAs with more than two clocks~\cite{alur1993parametric}. 
Significant progress could only be made in 2002  when the  subclass of L/U PTA  were  proposed in \cite{HUNE2002183} and the emptiness problem was 
proved to be decidable for these automata. In the following, we will extend these results and define some classes of PTAs for which  we propose solutions
 to the parameter synthesis problem and the emptiness problem.

\section{Parametric timed automata with one parametric clock} 
\label{sec:newresult1}

In this section we consider {\em parameter synthesis problem} of PTA with one parametric clock and   arbitrarily many  
parameters. The time values $\TT=\NN$ and parameter values $\PP=\RRR$. We first provide some result of \CAD, then
 prove the synthesis problem of PTA with one parametric clock is solvable.

\subsection{Cylindrical Algebraic Decomposition}
Delineability plays a crucial role in the theory of \CAD. Following the terminology used in \CAD, we say a
connected subset of $\RRR^m$ is a region. Given a region $S$, the cylinder $Z$ over $S$ is $S\times \RRR$.
A $\theta$-section of $Z$ is  a set of points $\left<\aaa, \theta(\aaa) \right>$, where $\aaa$ is in $S$ and $\theta$ 
is continuous function from $S$ to $\mathbb{R}$. A $(\theta_1, \theta_2)$-sector of $Z$ is the set of points $\left<\aaa,\beta\right>$,
where $\aaa$ is in $S$ and $\theta_1(\aaa)< \beta< \theta_2(\aaa)$ for continuous functions $\theta_1< \theta_2$ from 
$S$ to $\mathbb{R}$.  Sections and sectors are also regions. Given a subset of $S$ of $\mathbb{R}^m$,
a  decomposition of $S$ is a finite collection of disjoint regions $S_1,\cdots, S_k$ such than $S_1\cup\cdots \cup S_k=S$. 
 Given a region $S$, and a set of continuous functions $\theta_1<\cdots<\theta_k$ from $S$ to
$\mathbb{R}$, we can decompose the cylinder $S\times \mathbb{R}$ into the following regions:
\begin{itemize}
	\item the $\theta_i$-sections, for $1\le i\le k$, and
	\item $(\theta_i,\theta_{i+1})$-sections, for $0\le i\le k$,
\end{itemize}
where, with sight abuse of notation, we define $\theta_0$ as the constant function that return $-\infty$ and $\theta_{k+1}$ the constant function that 
return  $\infty$.
A set of polynomials  $\{f_1,\cdots, f_s\}\subset \ZZ[P,x],$ $P=[p_1,\cdots,p_{m}]$, is  said to be delineable in a region $S\subset \mathbb{R}^{m-1}$ if the following conditions hold:
\begin{enumerate}
	\item For every $1\le i \le s$, the total number of complex roots of $f_i(\aaa, x)$  remains invariant for any $\aaa\in S$.
	\item For every $i\le i\le s$, the number of distinct  complex roots of $f_i(\aaa,x)$ remains 
	invariant for any $\aaa$ in $S$.
	\item  For every $1\le i<j \le s$, the number of common complex roots of $f_i(\aaa, x)$ and $f_j(\aaa,x)$
	remains invariant for any $\aaa$ in $S$.
\end{enumerate}

A {\em sign assignment} for a set of polynomials $F$ is a mapping $\delta$, from  polynomials 
in $F$ to $\{-1,0,1\}$. Given  a set  of polynomials $F\subset \ZZ[P,p]$, we say a sign assignmemnt $\delta$ is {\em realizable}  with  respect to some
 $\aaa$ in $\mathbb{R}^m$ , if there exists a $\beta\in \mathbb{R}$ such that every $f\in F$ takes the sign corresponding to its sign assignment, i.e.,  
\sgn$(f(\aaa,\beta))=\delta(f)$.  The function \sgn\ maps a real number to its sign $\{-1,0,1\}$.
 We use $\signs(F,\aaa)$ to denote the set of realizable sign assignments of $F$ with respect to $\aaa$.
 
 \begin{theorem}[Lemma 1 of \cite{jovanovic2013solving}]
 	If a set of polynomials $F\subset \ZZ[P,x]$ is delineable  over a region $S$, then $\signs(F,\aaa)$ is invariant over $S$.
 \end{theorem}

 \begin{theorem}[Main algorithm of \cite{collins1}]
 	\label{the:cad}
 	$F$ is a set of polynomials in $\ZZ[P,x]$,  there is a algorithm  which computes  decomposition  $\RRR^m$  $S_1,\cdots,S_k$
 	such that $F$ is delineable over $S_i$ for $i=1,\cdots,k$. 
 \end{theorem}

 \begin{lemma}
 \label{lem:cad}
 For a polynomials formula $\phi$ where each polynomial of $\phi$ in $\ZZ[P,x]$, there is a  decomposition $S_1,\cdots,S_k$ of $\RRR^m$
 such that $\phi$ is true or false for each point of $S_i$ for $i=1,\cdots,k$. Moreover, \CAD\  provides a 
 sample point $\aaa_i$ where $\aaa_i\in S_i$ for $i=1,\cdots, k$. 
 \end{lemma}

\subsection{Parametric timed automata with one parametric clock}
The establishment and proof of this theorem involve a sequence of techniques to reduce the problem to computing the set of reachable states of an LTS. 
The major steps of  reduction include 
\begin{enumerate}
\item Reduce the problem of satisfaction of a system property $\psi$, say in the form of $\exists \Diamond \phi$,  by a  run $\tau$ to a reachability problem. 
This is done by encoding the state property in $\psi$ as a conjunction of  the invariant of a state.
\item Then we   move the  state invariants in a run out of the states and conjoin them to the guards of the corresponding transitions. 
\item Construct feasible runs for a given syntactic run in order to reach a given location.  This requires to define the notions of lower and  upper bounds of 
guards of transitions, through which an lower bound of feasible  parameter valuation is defined. 
\end{enumerate} 

 \subsection{Reduce satisfaction of  system to  reachability problem} 
We note  that $\psi$ is either of the form $\exists \Diamond \phi$ or  the  dual form $\forall  \Box \phi$, where $\phi$ is a state property.  Therefore, we only 
need to consider the problem of computing the set $\Gamma(\TA,\psi)$  for the case when $\psi$ is a formula of the form  $\exists \Diamond \phi$, i.e., 
there is a syntactic run $\tau$ such that  $\tau[\gamma] \models \phi$ for every $\gamma\in \Gamma(\TA,\psi)$. Our idea is to reduce the problem  of 
deciding $\TA \models \psi$ to a reachability problem of an LTS by encoding the state property $\phi$ in  $\exists \Diamond \phi$  into the guards of the 
transitions of $\mathcal{A}$. 
 
\begin{definition}[Encoding state property]
	Let  $\phi$ be a  state formula and $q$ be a location. We definite $\alpha(\phi,q)$ as follows, where $\equiv$ is used to denote syntactic equality 
	between formulas:
	\begin{itemize}
		\item  $\alpha(\phi, q)\equiv \phi$ if $\phi \equiv x-y\prec e$,  $\phi \equiv x\prec e$ or   $\phi \equiv -x\prec e$,  where $x$ and $y$ are clocks 
		and $e$ is an expression. 
		\item  when $\phi$ is a location $q'$,  $\alpha(\phi,q')\equiv true$ if $q'$ is $q$ and  $false$ otherwise.
		\item $\alpha $ preserves all Boolean connectives, that is   $\alpha(\neg \phi_1,q)\equiv \neg\alpha(\phi_1,q)$, 
		 $\alpha(\phi_1\wedge\phi_2,q)\equiv \alpha(\phi_1,q)\wedge \alpha(\phi_2,q)$, and  
		 $\alpha(\phi_1\vee \phi_2,q)\equiv  \alpha(\phi_1,q)\vee \alpha(\phi_2,q)$.
	\end{itemize}
\end{definition}
We can easily prove the   following lemma.

\begin{lemma}
	\label{lem:moveP}
	Given a PTA\ $\TA$, $\psi\equiv \exists\Diamond \phi$,  and  a syntactic run of $\mathcal{A}$ 
	\[\tau=(q_0,I_{q_0})\xrightarrow{g_1\& a_1[u_1]}(q_1,I_{q_1}){\cdots} \xrightarrow{g_\ell\&a_\ell[u_\ell]}(q_\ell,I_{q_\ell})\] 
we overload the function notation $\alpha$ and define the encoded run $\alpha(\tau)$ to be 
		\[
	 (q_0,I_{q_0})\xrightarrow{g_1\&a_1[u_1]}(q_1{,}I_{q_1}){\cdots}\xrightarrow{g_\ell\&a_\ell[u_\ell]}(q_\ell{,}I_{q_\ell }{\wedge} \alpha(\phi{,}q_\ell))
\]
	Then $\tau$ satisfies $\psi$ under parameter valuation $\gamma$ if and only if $R(\mathcal{A}_{\alpha (\tau)}[\gamma])\neq \emptyset$.
\end{lemma}
 Notice the  term guard is slightly abused in the lemma as $\alpha(\phi,q_\ell)$  may have disjunctions, and thus it may not be a simple constraint.
 
 \subsection{Moving state  invariants to guards of transitions}
  It is easy to see that both the invariant $I_q$  in the pre-state of the transition and the guard $g$ in a transition step  $(q,I_q)\xrightarrow{g\&a[u]} (q',I_{q'})$
   are both enabling conditions for the transition to take place.  Furthermore, the invariant $I_{q'}$  in the post-state of a transition needs to be guaranteed by 
   the set of clock resets $u$. Thus we can also understand this constraint as a guard condition for the transition to take place (the transition is not allowed to 
   take place if the invariant of the post-state is false. 
  
 For a PTA $\TA$ and a syntactic  run
	 \[\tau=(q_0,I_{q_0})\xrightarrow{g_1\&a_1[u_1]}(q_1{,}I_{q_1}){\cdots}\xrightarrow{g_\ell \&a_\ell[u_\ell]}(q_\ell ,I_{q_\ell}).\] 
	 Let $\overline{g}_i = (g_i\wedge I_{q_{i-1}}\wedge I_{q_i}[u_i])$. We define $\beta(\tau)$ as \[
		\begin{split}
	(q_0,true)\xrightarrow{\overline{g}_1\&a_1[u_1]}(q_1,true)\cdots
	\xrightarrow{\overline{g}_\ell \&a_\ell[u_\ell]}(q_\ell, true)
	\end{split}
\]	
\begin{lemma}	
\label{lem:moveI}
	For a PTA $\TA$, parameter valuation $\gamma$ and  a syntactic  run
	 \[\tau=(q_0,I_{q_0})\xrightarrow{g_1\&a_1[u_1]}(q_1{,}I_{q_1}){\cdots}\xrightarrow{g_\ell \&a_\ell[u_\ell]}(q_\ell ,I_{q_\ell})\] 
we have $(\gamma, (0,\cdots, 0))\models I_{q_0}$ and $R(\mathcal{A}_{\beta(\tau)}[\gamma])\neq \emptyset$ if and only if 
 $R(\mathcal{A}_\tau[\gamma])\neq \emptyset $.
\end{lemma}
\begin{proof}
Assume $(\gamma, x=0)\models I_{q_0}$ and $R(\mathcal{A}_{\beta(\tau)}[\gamma])\neq \emptyset$.
There is run $\xi $ of $\mathcal{A}_{\beta(\tau)}[\gamma]$  which is an alternating sequence of instantaneous and time advance transition steps 
  \[\xi= (q_0{,}\omega_0)\xrightarrow{d_0} (q_0{,}\omega_0') \xrightarrow{a_1} (q_1{,}\omega_1) \cdots   \xrightarrow{a_\ell} (q_\ell{,}\omega_\ell)\]
  such that $(\gamma, \omega_i')\models g_{a_{i+1}}\wedge I_{q_{i}}\wedge I_{q_{i+1}}[u_{a_{i+1}}]$ and $\omega_{i+1}=\omega_i'[u_{a_i}]$
 for $i=0,\cdots, \ell-1$.  Hence,  by the definition of $\mathcal{A}_\tau[\gamma]$, $\xi$ is also a run of $\tau$ under $\gamma$, and thus
  $R(\mathcal{A}_\tau[\gamma])\neq \emptyset$. 
 
  For the ``if'' direction,  assume there is $\xi$ as defined above which is  a  run of $\tau$ for the  parameter valuation  $\gamma$.
 Then by the  definition of the concrete semantics,  we have $(\gamma, x=0)\models I_{q_0}$, $(\gamma, \omega_i')\models g_{a_{i+1}} \wedge I_{q_{i}} $
  and $(\gamma, \omega_i'[u_{a_{i+1}}])\models I_{q_{i+1}} $ for $i=0.\cdots, \ell-1$.	
  	 In other words,  $(\gamma, \omega_i')\models I_{q_{i+1}}[u_{a_{i+1}}] $ for $i=0.\cdots, \ell-1$.
  	 	 Therefore,  $(\gamma, (0,\cdots,0))\models I_{q_0}$ and
 $\xi$ is a run of $\beta(\tau)$ under $\gamma$, i.e., $R(\mathcal{A}_{\beta(\tau)}[\gamma])\neq \emptyset$. \qed
\end{proof}

Since there is one parametric clock $x$, we can divide the conjuncts of simple constraint $g$ into two parts $\ELB(g)$ and $\EUP(g)$ where $g=\ELB(g)\wedge \EUP(g)$ and every conjunct of $\ELB(g)$ with form $-x\prec e$, every  conjunct of $\EUP(g)$ with form $x\prec e$.

\begin{definition}
\label{def:min}
For a concrete constraint $g$ we use $\linf(g)$ to denote the  infimum nonnegative value which satisfies $\ELB(g)$, if there is no value which 
 makes $\ELB(g)$ satisfy then $\linf(g)=\infty$. And we use $\usup(g)$ to denote the supremum nonnegative value which satisfies $\EUP(g)$, if there is no value which 
 makes $\EUP(g)$ satisfy then $\usup(g)=0$.
\end{definition}

\begin{definition}
 For a syntactic run \[\tau=(q_0,true)(g_1,a_1,\emptyset)(q_1,true)\cdots (g_\ell,a_\ell,\emptyset)(q_\ell,true)\]  with one clock $x$ in PTA $\TA$ where
	 $q_i\in Q, (q_{i-1},g_{i},a_i,\emptyset,q_i) \in \rightarrow$ and a parameter valuation $\gamma$, we use  $\varphi_{i,j}(\tau,\gamma)$ denote formula 
	 \[\left( \left(x\ge 0\right) \wedge \ELB(g_i[\gamma])\wedge   \EUP(g_j[\gamma]) \right) .\]

\end{definition}

\begin{lemma}
	\label{lem:withoutreset}
 	 For a syntactic run \[\tau=(q_0,true)(g_1,a_1,\emptyset)(q_1,true)\cdots (g_\ell,a_\ell,\emptyset)(q_\ell,true)\]  with one clock $x$ in PTA $\TA$ where
	 $q_i\in Q, (q_{i-1},g_{i},a_i,\emptyset,q_i) \in \rightarrow$, $R(\TA_{\tau}[\gamma])\neq \emptyset $  under  parameter valuation $\gamma$  if and 
	 only if formula $\varphi_{i,j}(\tau,\gamma)$  	is satisable for $j=i,\cdots,\ell,\ i=1,\cdots, \ell$.
 	 
\end{lemma}
\begin{proof}
	The ``if" side is easy to check.
	For prove ``only if" side,
	let 
	\[
	\omega_i=\frac{\max\{\linf(g_j[\gamma])\mid j=1,\cdots, i \} +\min\{\usup(g_j[\gamma])\mid j=i,\cdots, \ell \} }{2}.	\]
	
 We claim that 
	\[\xi=s_0  d_0 s_0' a_1 s_1\cdots  d_{\ell-1}s_{\ell-1}' a_\ell s_\ell \] is a run of $\TA$ where 
	$s_0=(q_0,0),\ d_i=\omega_{i+1}-\omega_i,\  s_i=(q_i, \omega_i),\ s_i'=(q_i, \omega_{i+1})$ for $i=0,\cdots, \ell$. 
	Since $\varphi_{i,j} $  is satisable,  $\linf(g_i[\gamma])\le \usup(g_j[\gamma]) $ if $j\ge i$. Hence $\linf(g_i[\gamma])\le \usup(g_i[\gamma]))$ for 
	$i=1,\cdots, \ell$.  Since
\[
	\omega_{i+1}=\frac{\max\{\linf(g_j[\gamma])\mid j=1,\cdots, i+1 \} +\min\{\usup(g_j[\gamma])\mid j=i+1,\cdots, \ell \} }{2},\]
	 $\max\{\linf(g_j[\gamma])\mid j=1,\cdots, i+1 \}\ge \max\{\linf(g_j[\gamma])\mid j=1,\cdots, i \}$ and 
	$\min\{\usup(g_j[\gamma])\mid j=i+1,\cdots, \ell \}  \ge \min\{\usup(g_j[\gamma])\mid j=i,\cdots, \ell \}  $. So, $\omega_{i+1} \ge \omega_i$.
		
	 Hence, for proving the claim we only need to prove that $\omega_i$  makes   constraint $g_i[\gamma]$ satisable for $i=1,\cdots, \ell$. As 
	 $\linf(g_i[\gamma])\le \usup(g_j[\gamma]) $ when $j\ge i$, 
	 \begin{equation}
	 \label{eq:1}
	 \begin{split}
	\omega_i&=\frac{\max\{\linf(g_j[\gamma])\mid j=1,\cdots, i \} +\min\{\usup(g_j[\gamma])\mid j=i,\cdots, \ell \} }{2}\\ 
	&\ge 	\frac{ \linf(g_i[\gamma]) + \min\{\usup(g_j[\gamma])\mid j=i,\cdots, \ell \}  }{2}\\
	&\ge 	\frac{\linf(g_i[\gamma]) + \linf(g_i[\gamma])}{2}\\
	&=\linf(g_i[\gamma])
	\end{split} 
    	\end{equation}
	and 
	\begin{equation}
	 \label{eq:2}
	 \begin{split}
	\omega_i&=\frac{\max\{\linf(g_j[\gamma])\mid j=1,\cdots, i \} +\min\{\usup(g_j[\gamma])\mid j=i,\cdots, \ell \} }{2}\\ 
	&\le 	\frac{\max\{\linf(g_j[\gamma])\mid j=1,\cdots, i \} +\usup(g_i[\gamma])   }{2}\\
	&\le 	\frac{\usup(g_i[\gamma]) + \usup(g_i[\gamma])}{2}\\
	&=\usup(g_i[\gamma]).
	\end{split} 
	\end{equation}
	   
	 We prove the claim by cases
	 
	 	\begin{itemize}
	 	\item  When   $\omega_i> \linf(g_i[\gamma])$ and  $\omega_i< \usup(g_i[\gamma])$. It is easy to check it this case,  formula $\omega_i\models g_i[\gamma]$ holds.
	 	\item When   $\omega_i= \linf(g_i[\gamma])$ and  $\omega_i< \usup(g_i[\gamma])$. 
	 	As the definition of equation (\ref{eq:1}), there is a $\usup(g_j[\gamma])=\omega_i$ and $j\ge i$. Since   $\varphi_{i,j} (\tau,\gamma)$ is satisable,  $\omega_i$
	 	is the only value which makes $\varphi_{i,j}(\tau,\gamma) $ hold. Hence, 
	 	 $\omega_i$ satisfies constraint $\ELB(g_i[\gamma])$. Moreover,  formula $\omega_i\models g_i[\gamma] $ holds.
	 	\item When   $\omega_i> \linf(g_i[\gamma])$ and  $\omega_i= \usup(g_i[\gamma])$. 
	 			As the definition of equation (\ref{eq:2}), there is a $\linf(g_j[\gamma])=\omega_i$ and $j\le i$. Since   $\varphi_{j,i} (\tau,\gamma)$ is  satisable and $\omega_i$
	 			is the only value which makes   $\varphi_{j,i}(\tau,\gamma) $ hold. Hence,  formula
	 			$\omega_i\models \EUP(g_i[\gamma])$ hold. Moreover,  formula $\omega_i\models g_i[\gamma]$ holds.
	 	\item When   $\omega_i= \linf(g_i[\gamma])$ and  $\omega_i= \usup(g_i[\gamma])$. 
	 		As the definition of equation (\ref{eq:1}), there is a $\usup(g_j[\gamma])=\omega_i$ and $j\ge i$. Since   $\varphi_{i,j}(\tau,\gamma) $ is  satisable and $\omega_i$
	 		is the only value which makes    $\varphi_{i,j}(\tau,\gamma) $ hold. Hence, 
	 		$\omega_i\models \ELB(g_i[\gamma])$.  	As the  definition of equation (\ref{eq:2}), there is a $\linf(g_j[\gamma])=\omega_i$ and $j\le i$. 
			Since   $\varphi_{j,i} $  is  satisable and $\omega_i$
	 		is the only value which make   $\varphi_{j,i} (\tau,\gamma)$ hold. Hence, formula
	 		$\omega_i\models \EUP(g_i[\gamma]) $ hols.  So,  formula $\omega_i\models g_i[\gamma]$ holds.
	  	\end{itemize}
	\qed
\end{proof}

{\em Lemma \ref{lem:withoutreset}} only solves the case when there is no update in the syntactic run. The following lemma will furtherly solve the case 
when there exists update.
\begin{lemma}
	\label{lem:withreset} 
	For a syntactic run \[\tau=(q_0,true)(g_1,a_1,u_1)(q_1,true)\cdots (g_\ell,a_\ell,u_\ell)(q_\ell,true)\] in PTA $\TA$ with one clock $x$ where 
	$q_i\in Q, (q_{i-1},g_{i},a_i,u_i,q_i) \in \rightarrow$,   $R(\TA_{\tau}[\gamma])\neq \emptyset $  under  parameter valuation $\gamma$  if and only if   formula
	$\varphi_{i,j}(\tau,\gamma)$ is  satisable  for $j=i,\cdots,\ell,\ i=1,\cdots, \ell$.
	
\end{lemma}
\begin{proof}
	The ``if" side is easy to check. Let $k$ be the number of transition which contain non-empty update set.  We proof
	by induction k. When $k=0$, by the {\em Lemma \ref{lem:withoutreset}}, the conclusion holds.
	
	Assume that the conclusion holds, when $k\le K$ where $K\ge 0$. When $k=K+1$, let $h$ be first index which contains 
	non-empty update set.  Let 
	 \[\tau_1=(q_0,true)(g_1,a_1,u_1)(q_1,true)\cdots (g_h,a_h,u_h)(q_h,true)\] and
	\[\tau_2=(q_0,true)(g_1,a_1,u_1)(q_1,true)\cdots (g_h,a_h,true)(q_h,true).\] As $\tau_2$ obtained 
	from relaxing the last update of $\tau_1$, 	employing {\em Lemma \ref{lem:withoutreset}} we have $R(\TA_{\tau_2}[\gamma])\neq \emptyset$.
	By the proof procedure of {\em Lemma  \ref{lem:withoutreset}}, there is a run
	\[\xi=s_0 d_0' s_0' a_1 s_1\cdots  d_{h-1}'s_{h-1}' a_h s_h \] of $\tau_2$ where
	 $s_0=(q_0,0),\ d_i'=\omega_{i+1}-\omega_i,\  s_i=(q_i, \omega_i),\ s_i'=(q_i, \omega_{i+1})$ for $i=0,\cdots, h$. After updating
	  $s_h=(q_h,\omega_h[u_h])$, $\xi$ is a run of $\tau_1$.	
	Let $\tau_3$ be a  syntactic run
	\[(q',true)(x= \omega_h[u_h],a,\emptyset )(q_{h},true)(g_{h+1},a_{h+1},u_{h+1})\cdots (g_\ell,a_\ell,u_{\ell})(q_\ell,true),\]
	where $x= \omega_i[u_h]$  is shorthand of $\left(x\le \omega_i[u_h]\right)\wedge \left(-x\le -\omega_i[u_h]\right) $. Since the number of transition 
	which contain non-empty update set is $K$ in $\tau_3$, by assumption,
	$R(\TA_{\tau_3}[\gamma])\neq \emptyset$. Combining $R(\TA_{\tau_1}[\gamma])\neq \emptyset$ and $R(\TA_{\tau_3}[\gamma])\neq \emptyset$,
	 we obtain $R(\TA_{\tau}[\gamma])\neq \emptyset$. Hence, the conclusion holds when $k=K+1$. Thus, the conclusion holds when $k\ge 0$.  
\qed	 
\end{proof}

\begin{theorem}[One parameter clock]
	\label{the:oneP}
	For a PTA $\TA$ with one parameter clock $x$ and  arbitrarily many  
	parameters,   set $\Gamma(\TA,\psi)$ is solvable if $\psi$ be  $\exists \Diamond \phi$.
\end{theorem}

\begin{proof}
Let $F$ be a set of polynomials which contains all constrained  polynomials occurring in $\TA$ and $\psi$.  
Employing {\em Lemma \ref{lem:cad}}, there is a decomposition $S_1,\cdots, S_k$ of $\RRR^{m}$ such that $\varphi$ is true 
or false in $S_i$ for $i=1,\cdots, k$ where $\varphi$  is a    combinations formula  of constraints occurring in $\TA$ and $\psi$. 
Moreover, \CAD\ provide a sample point $\aaa_i$ where $\aaa_i\in S_i$ for $i=1,\cdots, k$.

We {\bf claim} that if $\left(\gamma\in S_i\right) \wedge\left( \TA[\gamma]\models \psi\right)$, then 
 $\TA[\gamma']\models \psi$ for each $\gamma'\in S_i$, where $i=1,\cdots, k$.
 
 The claim can be proved as follows:
 
 Since $\left(\gamma\in S_i\right) \wedge\left( \TA[\gamma]\models \psi\right)$, there is a syntactic run 
  \[\tau=(q_0,I_{q_0})(g_{a_1},a_1,u_{a_1})(q_1,I_{q_1})\cdots (g_{a_\ell},a_\ell,u_{a_\ell})(q_\ell,I_{q_\ell})\]  where 
  $q_i\in Q, (q_{i-1},g_{a_i},a_i,u_{a_i},q_i) \in \rightarrow$ such that $
  \tau[\gamma]\models \psi$. Employing  {\em Lemma~\ref{lem:moveP}} and {\em Lemma \ref{lem:moveI}}, there is a
   a syntactic run 
   	\begin{equation*}
   	\begin{split}
     \tau_1=(q_0,true)(g_1,a_1,u_{a_1})(q_1,true)\cdots (g_\ell,a_\ell,u_{a_\ell})(q_\ell,true) 
    \end{split}
    \end{equation*} such that $\tau[\gamma]\models\psi$ if and only if
    $R(\TA_{\tau_1}[\gamma])\neq \emptyset$ for each $\gamma$. By {\em Lemma \ref{lem:withreset}}, 
     $R(\TA_{\tau_1}[\gamma])\neq \emptyset $  under  parameter valuation $\gamma$  if and only if  formula $ \varphi_{i,j}(\tau_1,\gamma)$ is satisfiable  for $j=i,\cdots,\ell,\ i=1,\cdots, \ell$. Since  $\varphi_{i,j}$ is constraint which combines with some constraints in $F$,   $\varphi_{i,j} (\tau_1,\gamma)$ is true or false in $S_h$ for $h=1,\cdots, k$.
     Hence,  if $R(\TA_{\tau_1}[\aaa_i])\neq \emptyset$, then  $R(\TA_{\tau_1}[\gamma])\neq \emptyset$ for each $\gamma\in S_i$, $i=1,\cdots, k$.
     Therefore, the claim holds, moreover the conclusion holds.
\qed	
\end{proof}

\begin{corollary}
	\label{coro:oneP}
		For a PTA $\TA$  with one parameter clock $x$ and  arbitrarily many  
		parameters,  set $\Gamma(\TA,\psi)$ is solvable if $\psi$ be  $\forall \Box \phi$.

\end{corollary}

\section{Parametric timed automata with linear   expression}
\label{sec:simple}

The {\em Theorem \ref{the:oneP}} is a beautiful result,  but it based one \CAD.   Thus,   even
with the improvements and various heuristics, \CAD's doubly-exponential worst-case
behavior has remained as a serious impediment. Hence,  for giving a practical algorithm, we  limit the expression which 
occurring in  PTA $\TA$ and property $\psi$ entirely contain in $\mathcal{LE}$.

In this section we consider synthesis problem of PTA  with one parameter clock $x$ and arbitrarily many  
parameters. The time values $\TT=\NN$ and parameter values $\PP=\ZZ$.

In this section, all the expressions are restricted to {\em linear   expression}. When there is a non-close constraint $e_1 < e_2$, as each 
variable take value in integer and each coefficient of $e_1,e_2$ are integer,    $e_1 < e_2$ is equivalent to  $e_1 \le  e_2-1$. So in the following 
we only consider close constraint.  

\begin{equation}
\label{eq:int}
\left\{ 	\begin{array}{@{}ll@{}}
a_{11}p_1+\cdots+a_{1n}p_m\ge b_1,\\
\vdots \\
a_{r1}p_1+\cdots+a_{rm}p_m\ge b_r,\\
a_{(r+1)1}p_1+\cdots+a_{(r+1)m}p_m=b_{r+1},\\
\vdots\\
a_{(r+t)1}p_1+\cdots+a_{(r+t)m}p_m=b_{r+t},
\end{array}       
 \right.
\end{equation}
with $a_{ij},b_j\in\ZZ$. In order to solve it we will use the following supplementary systems of linear Diophantine equations:
\begin{equation}
\label{eq:dio}
\left\{ 	\begin{array}{@{}ll@{}}
a_{11}p_1+\cdots+a_{1n}p_m-p_{m+1}= b_1,\\
\vdots \\
a_{r1}p_1+\cdots+a_{rm}p_m-p_{m+r}= b_r,\\
a_{(r+1)1}p_1+\cdots+a_{(r+1)m}p_m=b_{r+1},\\
\vdots\\
a_{(r+t)1}p_1+\cdots+a_{(r+t)m}p_m=b_{r+t}.
\end{array}       
\right.
\end{equation}

The  variables $p_{m+1},\dots, p_{m+r}$ are usually known in the literature as slack variables. There are many works about providing 
algorithm to solve solution of equation (\ref{eq:dio}) \cite{domich1987hermite,hochbaum1994can}.

 \begin{lemma}
 	\label{lem:lincad}
 	 Let  $\phi$ be a  a linear formula where each constraint of $\phi$ is form $e_1\prec e_2 $ where $e_1,e_2$ are linear expression in $\ZZ[p_1,\cdots,p_m,x]$.
	  There is a  decomposition $S_1,\cdots,S_k$ of $\RRR^m$ which each element $S_i$ can be presented as  form of equation (\ref{eq:int})
 	such that $\phi$ is true or false for each point of $S_i$ for $i=1,\cdots,k$. Moreover, \CAD\  provides a 
 	sample point $\aaa_i$ where $\aaa_i\in S_i$ for $i=1,\cdots, k$. 
 \end{lemma}

 \begin{theorem}[One parameter clock]
 	\label{the:oneP1}
 	For a PTA $\TA$with one parameter clock $x$ andarbitrarily many  
 parameters,  set $\Gamma(\TA,\psi)$ is solvable if $\psi$ be  $\exists \Diamond \phi$.
 \end{theorem}
 
 \begin{corollary}
 	\label{coro:oneP1}
 	 	For a PTA $\TA$with one parameter clock $x$ and arbitrarily many  
 	 	parameters,   set $\Gamma(\TA,\psi)$ is solvable if $\psi$ be   $\forall \Box \phi$.
 \end{corollary}

\bibliographystyle{acm}
\bibliography{ptm}

\begin{thebibliography}{10}

\bibitem{ptadai}


\bibitem{alur1999timed}
{\sc Alur, R.}
\newblock Timed automata.
\newblock In {\em Computer Aided Verification\/} (1999), Springer, pp.~8--22.

\bibitem{alur2002reachability}
{\sc Alur, R., Dang, T., and Ivan{\v{c}}i{\'c}, F.}
\newblock Reachability analysis of hybrid systems via predicate abstraction.
\newblock In {\em Hybrid Systems: Computation and Control}. Springer, 2002,
  pp.~35--48.

\bibitem{Alur90}
{\sc Alur, R., and Dill, D.}
\newblock Automata for modeling real-time systems.
\newblock In {\em Automata, Languages and Programming\/} (1990), Springer
  Berlin Heidelberg, pp.~322--335.

\bibitem{alur1994a}
{\sc Alur, R., and Dill, D.~L.}
\newblock A theory of timed automata.
\newblock {\em Theoretical Computer Science 126}, 2 (1994), 183--235.

\bibitem{alur2001parametric}
{\sc Alur, R., Etessami, K., La~Torre, S., and Peled, D.}
\newblock Parametric temporal logic for “model measuring”.
\newblock {\em ACM Transactions on Computational Logic (TOCL) 2}, 3 (2001),
  388--407.

\bibitem{alur1993parametric}
{\sc Alur, R., Henzinger, T.~A., and Vardi, M.~Y.}
\newblock Parametric real-time reasoning.
\newblock In {\em Proceedings of the twenty-fifth annual ACM symposium on
  Theory of computing\/} (1993), ACM, pp.~592--601.

\bibitem{Andr16}
{\sc Andr{\'e}, {\'E}.}
\newblock What's decidable about parametric timed automata?
\newblock In {\em Formal Techniques for Safety-Critical Systems\/} (2016),
  Springer International Publishing, pp.~52--68.

\bibitem{andre2009inverse}
{\sc Andr{\'e}, {\'E}., Chatain, T., Fribourg, L., and Encrenaz, E.}
\newblock An inverse method for parametric timed automata.
\newblock {\em International Journal of Foundations of Computer Science 20}, 05
  (2009), 819--836.

\bibitem{andre2015language}
{\sc Andr{\'e}, {\'E}., and Markey, N.}
\newblock Language preservation problems in parametric timed automata.
\newblock In {\em International Conference on Formal Modeling and Analysis of
  Timed Systems\/} (2015), Springer, pp.~27--43.

\bibitem{annichini2000symbolic}
{\sc Annichini, A., Asarin, E., and Bouajjani, A.}
\newblock Symbolic techniques for parametric reasoning about counter and clock
  systems.
\newblock In {\em Computer Aided Verification\/} (2000), Springer,
  pp.~419--434.

\bibitem{bandini2001application}
{\sc Bandini, G., Spelberg, R., de~Rooij, R.~C., and Toetenel, W.}
\newblock Application of parametric model checking-the root contention
  protocol.
\newblock In {\em System Sciences, 2001. Proceedings of the 34th Annual Hawaii
  International Conference on\/} (2001), IEEE, pp.~10--pp.

\bibitem{bengtsson2004timed}
{\sc Bengtsson, J., and Yi, W.}
\newblock Timed automata: Semantics, algorithms and tools.
\newblock In {\em Lectures on Concurrency and Petri Nets}. Springer, 2004,
  pp.~87--124.

\bibitem{bozzelli2009decision}
{\sc Bozzelli, L., and La~Torre, S.}
\newblock Decision problems for lower/upper bound parametric timed automata.
\newblock {\em Formal Methods in System Design 35}, 2 (2009), 121.

\bibitem{brown}
{\sc Brown, C.~W.}
\newblock Improved projection for cylindrical algebraic decomposition.
\newblock {\em J. Symb. Comput. 32}, 5 (2001), 447--465.

\bibitem{bruyere2007realtime}
{\sc Bruyere, V., and Raskin, J.}
\newblock Real-time model-checking: Parameters everywhere.
\newblock {\em Logical Methods in Computer Science 3}, 1 (2007), 1--30.

\bibitem{bundala2014advances}
{\sc Bundala, D., and Ouaknine, J.}
\newblock Advances in parametric real-time reasoning.
\newblock In {\em International Symposium on Mathematical Foundations of
  Computer Science\/} (2014), Springer, pp.~123--134.

\bibitem{collins1}
{\sc Collins, G.~E.}
\newblock Quantifier elimination for real closed fields by cylindrical
  algebraic decompostion.
\newblock In {\em Automata Theory and Formal Languages 2nd GI Conference
  Kaiserslautern, May 20--23, 1975\/} (1975), Springer, pp.~134--183.

\bibitem{Caviness1998}
{\sc Collins, G.~E.}
\newblock Quantifier elimination by cylindrical algebraic decomposition--twenty
  years of progress.
\newblock In {\em Quantifier elimination and cylindrical algebraic
  decomposition}. Springer, 1998, pp.~8--23.

\bibitem{Dolzmann2004}
{\sc Dolzmann, A., Seidl, A., and Sturm, T.}
\newblock Efficient projection orders for cad.
\newblock In {\em Proc. ISSAC'2004\/} (2004), ACM, pp.~111--118.

\bibitem{domich1987hermite}
{\sc Domich, P.~D., Kannan, R., and Trotter, E.~L.}
\newblock Hermite normal form computation using modulo determinant arithmetic.
\newblock {\em Mathematics of Operations Research 12}, 1 (1987), 50--59.

\bibitem{frehse2008counterexample}
{\sc Frehse, G., Jha, S.~K., and Krogh, B.~H.}
\newblock A counterexample-guided approach to parameter synthesis for linear
  hybrid automata.
\newblock In {\em Hybrid Systems: Computation and Control}. Springer, 2008,
  pp.~187--200.

\bibitem{han2014constructing}
{\sc Han, J., Dai, L., and Xia, B.}
\newblock Constructing fewer open cells by gcd computation in cad projection.
\newblock {\em international symposium on symbolic and algebraic computation\/}
  (2014), 240--247.

\bibitem{hochbaum1994can}
{\sc Hochbaum, D.~S., and Pathria, A.}
\newblock Can a system of linear diophantine equations be solved in strongly
  polynomial time?
\newblock {\em Citeseer\/} (1994).

\bibitem{Hong}
{\sc Hong, H.}
\newblock An improvement of the projection operator in cylindrical algebraic
  decomposition.
\newblock In {\em Proc. ISSAC'1990\/} (1990), ACM, pp.~261--264.

\bibitem{HUNE2002183}
{\sc Hune, T., Romijn, J., Stoelinga, M., and Vaandrager, F.}
\newblock Linear parametric model checking of timed automata.
\newblock {\em The Journal of Logic and Algebraic Programming 52-53\/} (2002),
  183 -- 220.

\bibitem{jovanovic2015integer}
{\sc Jovanovi{\'c}, A., Lime, D., and Roux, O.~H.}
\newblock Integer parameter synthesis for real-time systems.
\newblock {\em IEEE Transactions on Software Engineering 41}, 5 (2015),
  445--461.

\bibitem{jovanovic2013solving}
{\sc Jovanovic, D., and De~Moura, L.~M.}
\newblock Solving non-linear arithmetic.
\newblock {\em ACM Communications in Computer Algebra 46\/} (2013), 104--105.

\bibitem{McCallum1}
{\sc McCallum, S.}
\newblock An improved projection operation for cylindrical algebraic
  decomposition of three-dimensional space.
\newblock {\em J. Symb. Comput. 5}, 1 (1988), 141--161.

\bibitem{McCallum2}
{\sc McCallum, S.}
\newblock An improved projection operation for cylindrical algebraic
  decomposition.
\newblock In {\em Quantifier Elimination and Cylindrical Algebraic
  Decomposition}. Springer, 1998, pp.~242--268.

\bibitem{Tarski51}
{\sc Tarski, A.}
\newblock A decision method for elementary algebra and geometry.
\newblock In {\em Quantifier Elimination and Cylindrical Algebraic
  Decomposition\/} (1998), B.~F. Caviness and J.~R. Johnson, Eds., Springer
  Vienna, pp.~24--84.

\end{thebibliography}

\end{document}